\RequirePackage{fix-cm}
\documentclass[fleqn,smallextended]{svjour3}       

\usepackage{amsmath,amssymb}

\usepackage{tikz}
\usetikzlibrary{calc,positioning,arrows}
\usepackage{graphicx}
\usepackage{float}

\newtheorem{fact}{Fact}

\def\BH{{\sc Half-Square Of $\cal B$}}
\def\ECC{{\sc Edge Clique Cover}}

 \journalname{Algorithmica}
\begin{document}

\title{Hardness and structural results for half-squares of restricted tree convex bipartite graphs\thanks{This paper is an extended version of the COCOON 2017 paper~\cite{LeL17}.}
}

\titlerunning{Half-squares of restricted bipartite graphs}        

\author{Hoang-Oanh Le \and
        Van Bang Le 
}

\institute{Hoang-Oanh Le \at
              Berlin, Germany \\
              \email{LeHoangOanh@web.de}
           \and
           Van Bang Le \at
            Universit\"at Rostock, Institut f\"ur Informatik, Rostock, Germany\\
            \email{van-bang.le@uni-rostock.de}
}

\date{Received: date / Accepted: date}

\maketitle

\begin{abstract}
Let $B=(X,Y,E)$ be a bipartite graph. A half-square of $B$ has one color class of $B$ as vertex set, say $X$; two vertices are adjacent whenever they have a common neighbor in $Y$. Every planar graph is a half-square of a planar bipartite graph, namely of its subdivision. 
Until recently, only half-squares of planar bipartite graphs, also known as map graphs (Chen, Grigni and Papadimitriou [STOC 1998, J. ACM 2002]), have been investigated, and the most discussed problem is whether it is possible to recognize these graphs faster and simpler than Thorup's $O(n^{120})$-time algorithm (Thorup [FOCS 1998]).  

In this paper, we identify the first hardness case, namely that deciding if a graph is a half-square of a balanced bisplit graph is NP-complete. (Balanced bisplit graphs form a proper subclass of star convex bipartite graphs.) For classical subclasses of tree convex bipartite graphs such as biconvex, convex, and chordal bipartite graphs, we give good structural characterizations of their half-squares that imply efficient recognition algorithms. As a by-product, we obtain new characterizations of unit interval graphs, interval graphs, and of strongly chordal graphs in terms of half-squares of biconvex bipartite, convex bipartite, and of chordal bipartite graphs, respectively. Good characterizations of half-squares of star convex and star biconvex bipartite graphs are also given, giving linear-time recognition algorithms for these half-squares.
\keywords{Half-square \and NP-hardness \and graph algorithm \and computational complexity \and graph classes}
 \subclass{68R10 \and 05C85 \and 68Q25} 
\end{abstract}

\section{Introduction}
The square of a graph $H$, denoted $H^2$, is obtained from $H$ by adding new edges between two distinct vertices whenever their distance is two. Then, $H$ is called a square root of $G=H^2$. Given a graph $G$, it is NP-complete to decide if $G$ is the square of some graph $H$ (\cite{MotwaniS}), even for a split graph $H$ (\cite{LauC}). 

Given a bipartite graph $B=(X,Y,E_B)$, the subgraphs of the square $B^2$ induced by the color classes $X$ and $Y$, $B^2[X]$ and $B^2[Y]$, are called the two {\em half-squares\/} 
of $B$ (\cite{ChenGP98,ChenGP02}). 

While not every graph is the square of a graph and deciding if a graph is the square of a graph is hard, {\em every\/} graph $G=(V,E_G)$ is half-square of a bipartite graph: if $B=(V,E_G,E_B)$ is the bipartite graph with $E_B=\{ve \mid v\in V, e\in E_G, v\in e\}$, then clearly $G=B^2 [V]$. So one is interested in half-squares of special bipartite graphs. Note that $B$ is the subdivision of $G$, hence every planar graph is half-square of a planar bipartite graph.   

Let $\cal B$ be a class of bipartite graphs. A graph $G=(V,E_G)$ is called  {\em half-square} of $\cal B$ 
if there exists a bipartite graph $B=(V,W,E_B)$ in $\cal B$ such that $G=B^2[V]$. Then, $B$ is called a {\em ${\cal B}$ half-root} of $G$. With this notion, the following decision problem arises.
 
\medskip\noindent
\fbox{
\begin{minipage}{.96\textwidth} 
\BH\\[.5ex]
\begin{tabular}{l l}
{\em Instance:}& A graph $G=(V,E_G)$.\\
{\em Question:}& Is $G$ half-square of a bipartite graph in $\cal B$, i.e., does there \\ 
               & exist a bipartite graph $B=(V,W,E_B)$ in $\cal B$ s.t. $G=B^2[V]$?\\
\end{tabular}
\end{minipage}
}

\medskip\noindent
In this paper, we discuss \BH\, for several restricted bipartite graph classes $\cal B$.

\smallskip\noindent
\textbf{Previous results and related work.}\, Half-squares of bipartite graphs have been introduced in~\cite{ChenGP98,ChenGP02} in order to give a graph-theoretical characterization of the so-called map graphs. A {\em map graph\/} is the (point-)intersection graph of simply-connected and interior-disjoint regions of the Euclidean plane. More precisely, a {\em map\/} of a graph $G=(V,E_G)$ is a function $\cal M$ taking each vertex $v\in V$ to a closed disc homeomorph ${\cal M}(v)$ (the regions) in the plane, such that all ${\cal M}(v)$, $v\in V$, are interior-disjoint, and two distinct vertices $v$ and $v'$ of $G$ are adjacent if and only if the boundaries of ${\cal M}(v)$ and ${\cal M}(v')$ intersect. A map graph is one having a map. It turns out that map graphs are exactly half-squares of planar bipartite graphs (\cite{ChenGP98,ChenGP02}). As we have seen at the beginning, every planar graph is a map graph but not the converse; map graphs may have arbitrary large cliques. As such, map graphs form an interesting and important graph class.  
%
%
The main problem concerning map graphs is to recognize if a given graph is a map graph. 
In \cite{Thorup}, Thorup shows that \textsc{Half-Square Of Planar}, that is, deciding if a graph is a half-square of a planar bipartite graph, can be solved in polynomial time \footnote{Thorup did not give the running time explicitly, but it is estimated to be roughly $O(n^{120})$ with $n$ being the vertex number of the input graph; cf.~\cite{ChenGP02}.}. Very recently, in \cite{MnichRS}, it is shown that \textsc{Half-Squares Of Outerplanar} and \textsc{Half-Square Of Tree} are solvable in linear time. 
Other papers deal with solving hard problems in map graphs include \cite{Chen,DemaineFHT,DemaineH,EickmeyerK17,FominLS}. Some applications of map graphs have been addressed in~\cite{ChenHK99}. The paper~\cite{Brandenburg} discussed a relation between map graphs and $1$-planar graphs, an interesting topic in graph drawing.

\smallskip\noindent
\textbf{Our results.}\, We identify the first class $\cal B$ of bipartite graphs for which \BH\, is NP-hard. Our class $\cal B$ is a subclass of the class of the bisplit bipartite graphs and of star convex bipartite graphs (all terms are given later). For some other subclasses of tree convex bipartite graphs, such as star convex and star biconvex, convex and biconvex, and chordal bipartite graphs, we give structural descriptions for their half-squares, that imply polynomial-time recognition algorithms:
\begin{itemize}
\item Recognizing half-squares of balanced bisplit graphs (a proper subclass of star convex bipartite graphs) is hard, even when restricted to co-bipartite graphs;
\item Half-squares of star convex and star biconvex can be recognized in linear time; 
\item Half-squares of biconvex bipartite graphs are precisely the unit interval graphs;
\item Half-squares of convex bipartite graphs are precisely the interval graphs;
\item Half-squares of chordal bipartite graphs are precisely the strongly chordal graphs.
\end{itemize}

\section{Preliminaries}
Let $G=(V,E_G)$ be a graph with vertex set $V(G)=V$ and edge set $E(G)=E_G$. 
A \emph{stable set} (a \emph{clique}) in $G$ is a set of pairwise
non-adjacent (adjacent) vertices. The complete graph on $n$ vertices, the complete bipartite graph with $s$ vertices in one color class and $t$ vertices in the other color class, the cycle with $n$ vertices are denoted $K_n, K_{s,t}$, and $C_n$, respectively. A $K_3$ is also called a {\em triangle\/}, a complete bipartite graph is also called a {\em biclique}, a complete bipartite graph $K_{1,n}$ is also called a \emph{star}.

The neighborhood of a vertex $v$ in $G$, denoted by $N_G(v)$, is the set of all vertices in $G$ adjacent to $v$; if the context is clear, we simply write $N(v)$. 
A {\em universal vertex} $v$ in $G$ is one with $N(v)=V\setminus\{v\}$, i.e., $v$ is adjacent to all other vertices in $G$.

For a subset $W\subseteq V$,
$G[W]$ is the subgraph of $G$ induced by $W$, and $G-W$ stands for $G[V\setminus W]$. We write $B=(X,Y,E_B)$ for bipartite graphs with a bipartition into stable sets $X$ and $Y$. For subsets $S\subseteq X$, $T\subseteq Y$ we denote $B[S,T]$ for the bipartite subgraph of $B$ induced by $S\cup T$. 

We will consider half-squares of the following well-known subclasses of bipartite graphs: Let $B=(X,Y,E_B)$ be a bipartite graph.

\begin{itemize}
\item $B$ is \emph{$X$-convex} if there is a linear ordering on $X$ such that, for each $y\in Y$, $N(y)$ is an interval in $X$. Being \emph{$Y$-convex} is defined similarly. $B$ is \emph{convex} if it is $X$-convex or $Y$-convex. $B$ is \emph{biconvex} if it is both $X$-convex and $Y$-convex. We write \textsc{Convex} and \textsc{Biconvex} to denote the class of convex bipartite graphs, respectively, the class of biconvex bipartite graphs.
\item $B$ is \emph{chordal bipartite} if $B$ has no induced cycle of length at least six. 
\textsc{Chordal Bipartite} stands for the class of chordal bipartite graphs.
\item $B$ is \emph{tree $X$-convex} if there exists a tree $T=(X,E_T)$ such that, for each $y\in Y$, $N(y)$ induces a subtree in $T$. Being \emph{tree $Y$-convex} is defined similarly. $B$ is \emph{tree convex} if it is tree $X$-convex or tree $Y$-convex. $B$ is \emph{tree biconvex} if it is both tree $X$-convex and tree $Y$-convex. When $T$ is a star, we also speak of \emph{star convex} and \emph{star biconvex} bipartite graphs.

\textsc{Tree Convex} and \textsc{Tree Biconvex} are the class of all tree convex bipartite graphs and all tree biconvex bipartite graphs, respectively, and \textsc{Star Convex} and \textsc{Star Biconvex} are the class of all star convex bipartite graphs and all star biconvex bipartite graphs, respectively.
\end{itemize}

It is known that \textsc{Biconvex} $\subset$ \textsc{Convex} $\subset$ \textsc{Chordal Bipartite} $\subset$ \textsc{Tree Biconvex} $\subset$ \textsc{Tree Convex}. All inclusions are proper; 
see~\cite{Spinrad,Liu} 
for more information on these graph classes. 

Given a graph $G$, we often use the following two kinds of bipartite graphs associated to $G$:

\begin{definition}
\label{def:incidence}
Let $G=(V,E_G)$ be an arbitrary graph.
\begin{itemize}
\item The bipartite graph $B=(V,E_G, E_B)$ with $E_B=\{ve\mid v\in V, e\in E_G, v\in e\}$ is the \emph{subdivision} of $G$.
\item Let ${\cal C}(G)$ denote the set of all maximal cliques of $G$. The bipartite graph $B=(V,{\cal C}(G),E_B)$ with $E_B=\{vQ\mid v\in V, Q\in {\cal C}(G), v\in Q\}$ is the \emph{vertex-clique incidence bipartite graph} of $G$.
\end{itemize}
\end{definition}

Note that the subdivision of a planar graph is planar, and subdivisions and vertex-clique incidence graphs of triangle-free graphs coincide.
 
\begin{proposition}\label{prop:subdivision-incidence}
Every graph is half-square of its vertex-clique incidence bipartite graph. More precisely, if $B=(V,{\cal C}(G), E_B)$ is the vertex-clique incidence bipartite graph of $G=(V,E_G)$, then $G=B^2[V]$. Similar statement holds for subdivisions.
\end{proposition}
\begin{proof}
 For distinct vertices $u, v\in V$, $uv\in E_G$ if and only if $u, v\in Q$ for some $Q\in {\cal C}(G)$, if and only if $u$ and $v$ are adjacent in $B^2[V]$. That is, $G=B^2[V]$. 
 \qed
\end{proof}


\section{Recognizing half-squares of balanced bisplit graphs is hard}\label{sec:hardness}


A graph $G = (V,E)$ is a {\em split graph\/} if there is a partition $V = Q \,\dot\cup\, S$ of its vertex set into a clique $Q$ and stable set $S$. Recall that  a biclique is a complete bipartite graph. Following the concept of split graphs, we call a bipartite graph {\em bisplit\/} if it can be partitioned into a biclique and a stable set. In this section, we show that \textsc{Half-Square Of Balanced Bisplit} is NP-hard. Balanced bisplit graphs form a proper subclass of bisplit graphs, and are defined as follows.

\begin{definition}\label{def:splitbip2}
A bipartite graph $B=(X,Y,E_B)$ is called {\em balanced bisplit\/} if it satisfies the following properties:
\begin{itemize}
\item[(i)] $|X| = |Y|$;
\item[(ii)] there is partition $X=X_1\,\dot\cup\, X_2$ such that $B[X_1,Y]$ is a biclique;
\item[(iii)] there is partition $Y=Y_1\,\dot\cup\, Y_2$ such that the edge set of $B[X_2,Y_2]$ is a perfect matching.
\end{itemize} 
Note that by {(i)} and {(iii)}, $|X_1|=|Y_1|$, and by {(ii)} and {(iii)}, every vertex in $X_1$ is universal in $B^2[X]$.
\end{definition}

In order to prove the NP-hardness of \textsc{Half-Square Of Balanced Bisplit}, we will reduce the following well-known NP-complete problem \textsc{Edge Clique Cover} to it.

\medskip\noindent
\fbox{
\begin{minipage}{.96\textwidth} 
\ECC\\[.5ex]
\begin{tabular}{l l}
{\em Instance:}& A graph $G=(V,E_G)$ and a positive integer $k$.\\
{\em Question:}& Do there exist $k$ cliques in $G$ such that each edge of $G$ is\\
               & contained in some of these cliques?\\
\end{tabular}
\end{minipage}
}

\medskip\noindent
\ECC\ is NP-complete \cite{Holyer,KouSW,Orlin}, even when restricted to 
 co-bipartite graphs \cite{LeP}. (A co-bipartite graph is the complement of a bipartite graph.)

\begin{theorem}\label{thm:balancedsplitbip}
\textsc{Half-Square Of Balanced Bisplit} is NP-complete, even when restricted to co-bipartite graphs.
\end{theorem}
\begin{proof} 
It is clear that \textsc{Half-Square Of Balanced Bisplit} is in NP, since guessing a bipartite-half root $B=(V,W,E_B)$ with $|W|=|V|$, verifying that $B$ is balanced bisplit, and that 
$G = B^2[V]$ can obviously be done in polynomial time. Thus, by reducing \ECC\ to \textsc{Half-Square Of Balanced Bisplit}, we will conclude that \textsc{Half-Square Of Balanced Bisplit} is NP-complete.

Let $(G=(V,E_G),k)$ be an instance of \ECC. Note that we may assume that $k\le |E_G|$, and that $G$ is connected and has no universal vertices. 
We construct an instance $G'=(V',E_{G'})$ of \textsc{Half-Square Of Balanced Bisplit} as follows: $G'$ is obtained from $G$ by adding a set $U$ of $k$ new vertices, $U=\{u_1,\ldots, u_k\}$, and all edges between vertices in $U$ and all edges $uv$ with $u\in U$, $v\in V$. Thus, $V'=V\cup U$, $G'[V]=G$ and the $k$ new vertices in $U$ are exactly the universal vertices of $G'$. Clearly, $G'$ can be constructed in polynomial time $O(k|V|)=O(|E_G|\cdot|V|)$, and in addition, if $G$ is co-bipartite, then $G'$ is co-bipartite, too. 
We now show that $(G,k)\in$ \ECC\ if and only if $G'\in$ \textsc{Half-Square Of Balanced Bisplit}. 

First, suppose that the edges of $G=(V,E_G)$ can be covered by $k$ cliques $Q_1,\ldots, Q_k$ in $G$. We are going to show that $G'$ is half-square of some balanced bisplit bipartite graph. 
Consider the bipartite graph $B=(V',W,E_B)$ (see also Figure~\ref{fig:B}) with 
\[W=W_1\cup W_2,\, \text{ where }\, W_1=\{w_1, \ldots, w_k\},\,\text{ and }\, W_2=\{w_v\mid v\in V\}. \]
In particular, $|V'|=|W|=k+|V|$. The edge set $E_B$ is as follows:
\begin{itemize}
\item $B$ has all edges between $U$ and $W$, i.e., $B[U,W]$ is a biclique, 
\item $B$ has edges $vw_v$, $v\in V$. Thus, the edge set of $B[V,W_2]$ forms a perfect matching, and 
\item $B$ has edges $vw_i$, $v\in V$, $1\le i\le k$, whenever $v\in V$ is contained in clique $Q_i$, $1\le i\le k$. 
\end{itemize}

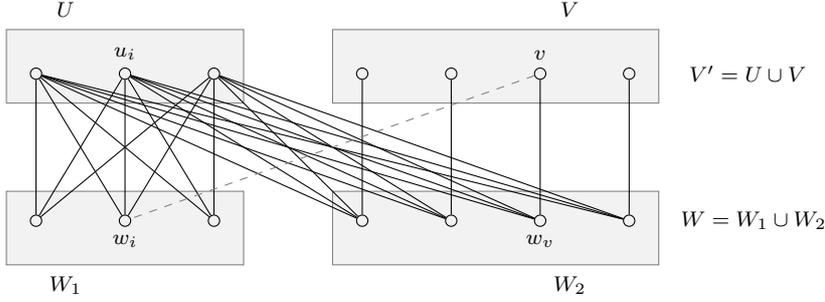
\begin{figure}[H] 
\begin{center}
\begin{tikzpicture}[scale=.39]
\tikzstyle{vertex}=[draw,circle,inner sep=1.5pt]; 

\filldraw[fill=black!5!white, draw=gray] (-1,4) rectangle (7,6.5); 
\node at (1,7.2) {$U$};
\node[vertex] (u1) at (0,5) {}; 
\node[vertex] (u2) at (3,5) [label=above:$u_i$] {}; 
\node[vertex] (u3) at (6,5) {}; 

\filldraw[fill=black!5!white, draw=gray] (10,4) rectangle (21,6.5); 
\node at (18,7.2) {$V$};
\node[vertex] (v1) at (11,5) {}; 
\node[vertex] (v2) at (14,5) {}; 
\node[vertex] (v3) at (17,5) [label=above:$v$] {};
\node[vertex] (v4) at (20,5) {};
\node at (24,5) {$V'=U\cup V$};

\filldraw[fill=black!5!white, draw=gray] (-1,-1.5) rectangle (7,1); 
\node at (1,-2.2) {$W_1$};
\node[vertex] (w1) at (0,0) {}; 
\node[vertex] (w2) at (3,0) [label=below:$w_i$] {}; 
\node[vertex] (w3) at (6,0) {}; 

\filldraw[fill=black!5!white, draw=gray] (10,-1.5) rectangle (21,1); 
\node at (18,-2.2) {$W_2$};
\node[vertex] (wv1) at (11,0) {}; 
\node[vertex] (wv2) at (14,0) {}; 
\node[vertex] (wv3) at (17,0) [label=below:$w_v$] {};
\node[vertex] (wv4) at (20,0) {};
\node at (24.2,0) {$W=W_1\cup W_2$};

\draw (v1) -- (wv1); \draw (v2) -- (wv2); \draw (v3) -- (wv3); \draw (v4) -- (wv4); 
\foreach \w in {wv1,wv2,wv3,wv4,w1,w2,w3}
{\foreach \u in {u1,u2,u3}
{
\draw (\w) -- (\u);
}
}
\draw[dashed, thin, gray] (v3) -- (w2);
\end{tikzpicture}
\end{center}
\caption{The balanced bisplit graph $B=(V',W,E_B)$ proving $G'=B^2[V']$; $v\in V$ is adjacent to $w_i\in W_1$ if and only if $v\in Q_i$.}\label{fig:B}
\end{figure}

Thus, $B$ is a balanced bisplit graph. Moreover, by the construction of $B$, we have in $B^2[V']$:
\begin{itemize}
\item $U=\{u_1,\ldots,u_k\}$ is a clique (as $B[U,W]$ is a biclique),  
\item every vertex $u\in U$ is adjacent to all vertices $v\in V$ (recall that $G$ is connected, so every $v\in V$ is in some $Q_i$, and $w_i\in W_1$ is a common neighbor of $u$ and $v$), and 
\item no two distinct vertices $v, z\in V$ have, in $B$, a common neighbor in $W_2$. So $u$ and $z$ are adjacent in $B^2[V']$ if and only if $v$ and $z$ have a common neighbor $w_i$ in $W_1$, if and only if $v$ and $z$ belong to a clique $Q_i$ in $G$, if and only if $u$ and $z$ are adjacent in $G$. 
\end{itemize}
That is, $G'=B^2[V']$.

\smallskip
Conversely, suppose $G'=H^2[V']$ for some balanced bisplit graph $H=(V',Y,E_H)$ with $|V'|=|Y|$ and partitions $V'=X_1\,\dot\cup\,X_2$ and $Y=Y_1\,\dot\cup\, Y_2$ as in Definition~\ref{def:splitbip2}. We are going to show that the edges of $G$ can be covered by $k$ cliques. As $H[X_1,Y]$ is a biclique, all vertices in $X_1$ are universal in $H^2[V']=G'$. Hence 
\[X_1=U\] 
because no vertex in $V=V'\setminus U$ is universal in $G'$ (recall that $G$ has no universal vertex). Therefore (recall that $G'[V]=G$) 
\[\text{$X_2=V$ and $G=H^2[V]$}.\] 
Note that, as $H$ is a balanced bisplit graph, $|Y_1|=|U|=k$. Write $Y_1=\{q_1,\ldots, q_k\}$ and observe that no two vertices in $V$ have a common neighbor in $Y_2$. Thus, for each edge $vz$ in $G=H^2[V]$, $v$ and $z$ have a common neighbor $q_i$ in $Y_1$. Therefore, the $k$ cliques $Q_i$ in $H^2[V]$, $1\le i\le k$, induced by the neighbors of $q_i$ in $V$, cover the edges of $G$.  
\qed
\end{proof}

Theorem~\ref{thm:balancedsplitbip} indicates that recognizing half-squares of restricted bipartite graphs is algorithmically much more complex than recognizing squares of bipartite graphs; the latter can be done in polynomial time (\cite{Lau}).

Observe that balanced bisplit graphs are star convex: Let $B=(X,Y,E_B)$ be a bipartite graph with the properties in Definition~\ref{def:splitbip2}. Fix a vertex $u\in X_1$ and consider the star $T=(X,\{uv\mid v\in X-u\})$. Since every vertex $y\in Y$ is adjacent to $u$, $N(y)$ induces a substar in $T$. Note, however, that the hardness of \textsc{Half-Square Of Balanced Bisplit} does not imply the hardness of \textsc{Half-Square Of Star Convex}. This is because the proof of Theorem~\ref{thm:balancedsplitbip} strongly relies on the properties of balanced bisplit graphs. 

Indeed, we will show in the next section that half-squares of star-convex bipartite graphs can be recognized in polynomial time. 


\section{Half-squares of star convex and star biconvex bipartite graphs}\label{sec:starconvex}


We need more notations for stating our results. Let $G=(V_G,E_G)$ and $H=(V_H,E_H)$ be two (vertex-disjoint) graphs. For a vertex $v\in V_G$, we say that the graph with 
\begin{itemize}
\item vertex set $(V_G\setminus\{v\})\cup V_H$ and 
\item edge set $(E_G\setminus\{e\in E_G\mid v\in e\})\cup E_H\cup\{ux\mid u\in N_G(v), x\in V_H\}$
\end{itemize}
is obtained from $G$ by substituting the vertex $v$ of $G$ by the graph $H$. Thus, {\em substituting\/} a vertex $v$ of $G$ by the graph $H$ results in the graph obtained from $G-v$ and $H$ by adding all edges between $N_G(v)$ and $V_H$. 

Recall that a split graph is one, whose vertex set can be partitioned into a clique and a stable set. In a graph, a connected component is {\em big\/} if it has at least two vertices.
\begin{lemma}\label{lem:star-convex-1}
Let $B=(X,Y,E_B)$ be a star convex bipartite graph with an associated star $T=(X,E_T)$. Then
\begin{itemize}
  \item[\em (i)] $B^2[X]$ has at most one big connected component and the big connected component has a universal vertex.
  \item[\em (ii)] $B^2[Y]$ is obtained from a split graph by substituting vertices by cliques.
\end{itemize}
\end{lemma}
\begin{proof} 
Let $x_0$ be the center vertex of the star $T=(X,E_T)$. 
Note that, if $N(y)$ has at least two vertices, then $N(y)$ must contain $x_0$.

\smallskip\noindent
(i): If $|N(y)|\le 1$ for all $y\in Y$, then $B^2[X]$ is clearly edgeless, and (i) trivially holds. So, assume that 
\[Y_0:=\{y\in Y\,:\, |N(y)|\ge 2\}\]
is not empty. Write
\[X_0:=N(Y_0)\,\text{ and }\, X_1:=X\setminus X_0.\]
Then $|X_0|\ge 2$ and $x_0\in X_0$. Moreover, for every $u\in X_1$ and $v\in X$, $u$ and $v$ have no common neighbor in $Y$. Thus, $B^2[X]$ consists of the big component induced by $X_0$ in which $x_0$ is a universal vertex and $|X_1|$ many one-vertex components.

\smallskip\noindent
(ii): Partition $Y$ into $Y_0$ and $Y_1$ with
\[Y_0:=\{y\in Y\mid x_0\in N(y)\}\,\text{ and }\, Y_1:=Y\setminus Y_0.\]
Then, clearly, 
\[\text{$Y_0$ is a clique in $B^2[Y]$}.\]
Let $Y_1\not=\emptyset$ (otherwise, (ii) obviously holds), and write $N(Y_1)=\{x_1,\ldots,x_p\}$ for some $p\ge 1$. Note that every vertex $y\in Y_1$ has degree one since $x_0\not\in N(y)$. Thus, $B[\{x_1,\ldots,x_p\},Y_1]$ consists of $p$ vertex-disjoint stars $(x_i,N(x_i)\cap Y_1)$ at center vertices $x_i$, $1\le i\le p$. For each $i$, fix a vertex $y_i\in N(x_i)\cap Y_1$. (See Figure~\ref{fig:lemma1ii}.)
\begin{figure}[H] 
\begin{center}
\begin{tikzpicture}[scale=.39]
\tikzstyle{vertex}=[draw,circle,inner sep=1.5pt]; 

\filldraw[fill=black!5!white, draw=gray] (-1,4) rectangle (21,6.5); 
\node[vertex] (u1) at (0,5) [label=above:$x_0$] {}; 
\node[vertex] (u2) at (3,5) {}; 
\node[vertex] (u3) at (6,5) {}; 
\node[vertex] (u4) at (8,5) {}; 

\node[vertex] (v1) at (12,5) [label=above:$x_1$] {}; 
\node[vertex] (v2) at (15.5,5) [label=above:$x_2$] {}; 
\node at (17.5,5) {$\cdots$};
\node[vertex] (vp) at (19.5,5) [label=above:$x_p$] {};
\node at (23,5) {$X$};

\filldraw[fill=black!5!white, draw=gray] (-1,-1.5) rectangle (7,1); 
\node at (1,-2.2) {$Y_0$};
\node[vertex] (w1) at (0,0) {}; 
\node[vertex] (w2) at (2,0) {}; 
\node[vertex] (w22) at (4,0) {}; 
\node[vertex] (w3) at (6,0) {}; 

\filldraw[fill=black!5!white, draw=gray] (10,-1.5) rectangle (21,1); 
\node at (18,-2.2) {$Y_1$};
\node[vertex] (y1) at (11,0) [label=below:$y_1$]{}; 
\node[vertex] (y11) at (12,0) {};
\node[vertex] (y111) at (13,0) {};

\node[vertex] (y2) at (15,0) [label=below:$y_2$]{}; 
\node[vertex] (y22) at (16,0) {}; 

\node[vertex] (yp) at (19,0) [label=below:$y_p$] {}; 
\node[vertex] (ypp) at (20,0) {};
\node at (23,0) {$Y$};

\draw (u1) -- (w1); \draw (u1) -- (w2); \draw (u1) -- (w22); \draw (u1) -- (w3); 
\draw (u2) -- (w1); \draw (u2) -- (w22); \draw (u3) -- (w22); \draw (u3) -- (w3); \draw (u4) -- (w3);  
 
\draw (v1) -- (w22); \draw (v1) -- (w3);
\draw (vp) -- (w3); 

\draw[very thick] (v1) -- (y1); \draw (v1) -- (y11); \draw (v1) -- (y111); 
\draw[very thick] (v2) -- (y2); \draw (v2) -- (y22);
\draw[very thick] (vp) -- (yp); \draw (vp) -- (ypp); 

\end{tikzpicture}
\end{center}
\caption{Proof of Lemma~\ref{lem:star-convex-1} (ii) illustrated.}\label{fig:lemma1ii}
\end{figure}
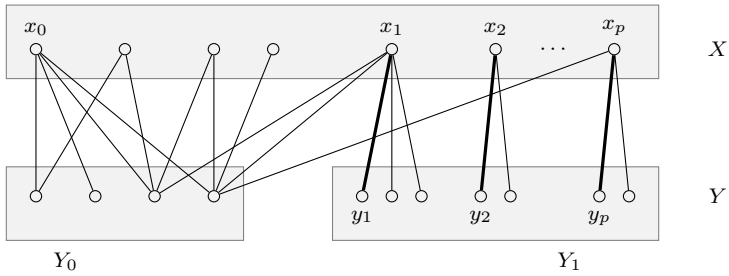
Then, as no two of $y_1,\ldots, y_p$ have a common neighbor in $X$, 
\[\text{$\{y_1,\ldots,y_p\}$ is a stable set in $B^2[Y]$}.\]
Thus, $B^2[Y]$ is obtained from the split graph $B^2[Y_0\cup\{y_1,\ldots,y_p\}]$ by substituting $y_i$ by clique with vertex set $N(x_i)\cap Y_1$, $1\le i\le p$.
\qed
\end{proof}

The following facts show that the reverse statements in Lemma~\ref{lem:star-convex-1} hold true, too.

\begin{fact}\label{fact:1}
Let $G=(V,E_G)$ have at most one big connected component, and let the big connected component have a universal vertex. Then $G$ is half-square of a star convex bipartite graph.
\end{fact}
\begin{proof}
It is obvious that graphs having no edges are half-squares of star convex bipartite graphs. 
So, let $U$ be the set of universal vertices of the big component of $G$, let $I$ be the set of isolated vertices of $G$, and let $R$ be set of the remaining vertices of $G$, $R=V\setminus (U\cup I)$. Thus, $G[U\cup R]$ is the big component of $G$. 

If $R=\emptyset$, let $B=(V,W,E_B)$ with $W=\{w\}$ and $E_B=\{wu\mid u\in U\}$.  
If $R\not=\emptyset$, construct a bipartite graph $B=(V, W, E_B)$ as follows.
\begin{itemize}
  \item $W=R'\cup E_R$, where $R'=\{x'\mid x\in R\}$ and $E_R= E_{G[R]}$ is the edge set (possibly empty) of $G[R]$,
  \item $B[U,R'\cup E_R]$ is a biclique, the edge set of $B[R,R']$ is the perfect matching $\{xx'\mid x\in R\}$,
  \item the remaining edges of $B$ are between $R$ and $E_R$. Two vertices $x\in R$ and $e\in E_R$ are adjacent in $B$ if and only if $x$ is an endvertex of the edge $e$ in $G[R]$. That is, $B[R,E_R]$ is the subdivision of $G[R]$.
\end{itemize}
It is not difficult to verify, by construction, that $G=B^2[V]$. Moreover, $B$ is star convex: Fix a vertex $u\in U$ and let $T=(V,E_T)$ be the star with edge set $E_T=\{uv\mid v\in V\setminus \{u\}\}$. Clearly, for every $w\in W$, $N(w)$ forms a substar in $T$.
\qed
\end{proof}

\begin{fact}\label{fact:2}
Graphs obtained from split graphs by substituting vertices by cliques are half-squares of star convex bipartite graphs.
\end{fact}
\begin{proof}
We first show that split graphs are half-squares of star convex bipartite graphs. 
Let $G=(V,E_{G})$ be a split graph with a partition of its vertex set $V=Q\cup S$ into clique $Q$ and stable set $S$. 
Construct a bipartite graph $H=(X,V, E_H)$ as follows.
\begin{itemize}
\item $X=\{x_0\}\cup\{x_s\mid s\in S\}$,
\item $\displaystyle E_H=\{x_0q\mid q\in Q\}\cup\{x_ss\mid s\in S\}\cup\bigcup_{s\in S} \{x_sq\mid q\in N_{G}(s)\}$.
\end{itemize}
By construction, $G=H^2[V]$. Moreover, $H$ is star convex with the associated star $T=(X,\{x_0x_s\mid s\in S\})$.

Now, if $G'=(V',E_{G'})$ is obtained from $G$ by substituting $v\in V$ by {\em clique\/} $Q_v$, then $G'$ clearly is the half-square $B^2[V']$ of the bipartite graph $B=(X,V',E_B)$ obtained from $H$ by substituting vertices $v$ by {\em stable sets\/} $Q_v$. Obviously, $B$ is star convex with the same star $T$ associated to $H$.   
\qed
\end{proof}

By Lemma~\ref{lem:star-convex-1} and Facts~\ref{fact:1} and~\ref{fact:2}, we obtain:
\begin{theorem}\label{thm:star-convex}
A graph is a half-square of a star convex bipartite graph if and only if  
\begin{itemize}
\item[\em (i)]  it has at most one big connected component and the big connected component has a universal vertex, or
\item[\em (ii)] it is obtained from a split graph by substituting vertices by cliques. 
\end{itemize}
\end{theorem}

In case of star biconvex bipartite graphs, we obtain:
\begin{theorem}\label{thm:star-biconvex}
A graph is a half-square of a star biconvex bipartite graph if and only if it has at most one big connected component and the big connected component is obtained from a split graph having a universal vertex by substituting vertices by cliques. 
\end{theorem}
\begin{proof}
The necessity part follows directly from Theorem~\ref{thm:star-convex}. 

For the sufficiency part, we first show that split graphs in which the big connected component has a universal vertex are half-squares of star biconvex bipartite graphs.   
Let $G=(V,E_{G})$ be a split graph with a partition of its vertex set $V=Q\cup S$ into clique $Q$ and stable set $S$. Let $I\subseteq S$ be the set of all isolated vertices  of $G$ (thus, $G-I$ is the big component of $G$).    
Construct a bipartite graph $H=(X,V, E_H)$ as follows. 
\begin{itemize}
\item $X=\{x_0\}\cup\{x_s\mid s\in S\setminus I\}\cup\{x_i\mid i\in I\}$,
\item $\displaystyle E_H=\{x_0q\mid q\in Q\}\cup\{x_ss\mid s\in S\setminus I\} \cup\{x_ii\mid i\in I\}\cup\bigcup_{s\in S\setminus I} \{x_sq\mid q\in N_{G}(s)\}$.
\end{itemize}
By construction, $G=H^2[V]$. Moreover, $H$ is star biconvex with the associated stars $T_1=(X,\{x_0x_s\mid s\in S\setminus I\}\cup \{x_0x_i\mid i\in I\})$ and $T_2=(V,\{uv\mid v\in V\setminus\{u\}\})$, where $u\in Q$ is a universal vertex of $G-I$ (hence $u\in N_G(s)$ for all $s\in S\setminus I$).

Now, if $G'=(V',E_{G'})$ is obtained from $G$ by substituting vertices $v\in V\setminus I$ by {\em cliques\/} $Q_v$, then $G'$ clearly is the half-square $B^2[V']$ of the bipartite graph $B=(X,V',E_B)$ obtained from $H$ by substituting vertices $v$ by {\em stable sets\/} $Q_v$. Obviously, $B$ is star biconvex with the same star $T_1$ associated to $H$ and the star $T_2'=(V',\{u'v'\mid v'\in V'\setminus\{u'\}\})$, where $u'$ is a vertex in $Q_u$ ($u$ is a universal vertex of $G-I$).  
\qed
\end{proof}

By definition, if $G$ is obtained from a graph $H$ by substituting vertex $v\in V(H)$ by clique $Q_v$ with $|Q_v|\ge 2$, then $Q_v$ is a {\em module\/} in $G$, that is, every vertex in $G$ outside $Q_v$ is adjacent to all or to none vertices in $Q_v$. Now, note that all maximal clique modules of a given graph can be computed in linear time (see, e.g, \cite[Corollary 7.4]{McConnell03}). Note also that split graphs can be recognized in linear time (cf.~\cite{Golumbic}), and a partition into a clique and a stable set of a given split graph can be computed in linear time (\cite{HeggernesK07}). Thus, Theorems~\ref{thm:star-convex} and \ref{thm:star-biconvex} and their proofs imply: 
\begin{corollary}\label{cor:starconvex}
\textsc{Half-Square Of Star Convex} and \textsc{Half-Square Of Star Biconvex} can be solved in linear time. A star (bi)convex bipartite half-root, if any, can be constructed in linear time. 
\end{corollary}


\section{Half-squares of biconvex and convex bipartite graphs}\label{sec:convex}

In this section, we show that half-squares of convex bipartite graphs are precisely the interval graphs and half-squares of biconvex bipartite graphs are precisely the unit interval graphs.

Recall that $G=(V,E_G)$ is an interval graph if it admits an interval representation $I(v), v\in V$, such that two vertices in $G$ are adjacent if and only if the corresponding intervals intersect. 
Let $G$ be an interval graph. It is well-known (\cite{FulkersonG,GilmoreH}) 
that there is a linear ordering of the maximal cliques of $G$, 
say ${\cal C}(G)=(Q_1,\ldots,Q_q)$,
such that every vertex of $G$ belongs to maximal cliques 
that are consecutive in that ordering,
that is, for every vertex $u$ of $G$, 
there are indices $\ell(u)$ and $r(u)$ with
\[\{ i\mid 1\leq i\leq q\mbox{ and }u\in Q_i\}=\{ i\mid \ell(u)\leq i\leq r(u)\}.\]

If $C$ and $D$ are distinct maximal cliques of $G$,
then $C\setminus D$ and $D\setminus C$ are both not empty,
that is, for every $j\in\{ 1,\ldots,q\}$, 
there are vertices $u$ and $v$ 
such that $r(u)=\ell(v)=j$.

Recall that unit interval graphs are those interval graphs admitting an interval representation in which all intervals have the same length. It is well known (\cite{Roberts}) that a graph is a unit interval graphs if and only if it has an interval representation in which no interval is properly contained in another interval (a proper interval graph), if and only if it is a $K_{1,3}$-free interval graph.

\begin{lemma}\label{lem:biconvex-claw}
The half-squares of a biconvex bipartite graph are $K_{1,3}$-free.
\end{lemma}
\begin{proof}
Let $B=(X,Y,E_B)$ be a biconvex bipartite graph. By symmetry we need only to show that $B^2[X]$ is $K_{1,3}$-free. 
Suppose, by contradiction, that $x_1, x_2, x_3, x_4$ induce a $K_{1,3}$ in $B^2[X]$ with edges $x_1x_2, x_1x_3$ and $x_1x_4$. Let $y_i$ be a common neighbor of $x_1$ and $x_2$, $y_j$ be a common neighbor of $x_1$ and $x_3$, and $y_k$ be a common neighbor of $x_1$ and $x_4$. Then, $y_i, y_j, y_k$ are pairwise distinct and induce, in $B$, a subdivision of $K_{1,3}$; see also Figure~\ref{fig:SK13}. This is a contradiction because biconvex bipartite graphs do not contain an induced subdivision of the $K_{1,3}$. Thus, the half-squares of a biconvex bipartite graph are $K_{1,3}$-free.
\qed
\end{proof}

\begin{figure}[ht]
\begin{center}
\begin{tikzpicture}[scale=.39]
\tikzstyle{vertex}=[draw,circle,inner sep=1.5pt]; 

\node[vertex] (x1) at (0,3) [label=above:$x_1$] {}; 
\node[vertex] (x2) at (3,3) [label=above:$x_2$] {}; 
\node[vertex] (x3) at (6,3) [label=above:$x_3$] {}; 
\node[vertex] (x4) at (9,3) [label=above:$x_4$] {}; 
\node[vertex] (yi) at (0,0) [label=below:$y_i$] {}; 
\node[vertex] (yj) at (3,0) [label=below:$y_j$] {};
\node[vertex] (yk) at (6,0) [label=below:$y_k$] {};

\draw (x1) -- (yi) -- (x2); \draw (x1) -- (yj) -- (x3); \draw (x1) -- (yk) -- (x4);

\end{tikzpicture}
\end{center}
\caption{The subdivision of $K_{1,3}$ is convex, but not biconvex.}\label{fig:SK13}
\end{figure}
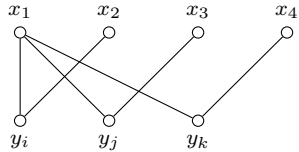

\begin{lemma}\label{lem:convex}
\mbox{}
\begin{itemize}
\item[\em (i)] Every interval graph is half-square of a convex bipartite graph. More precisely, if $G=(V,E_G)$ is an interval graph and $B=(V,{\cal C}(G), E_B)$ is the vertex-clique incidence bipartite graph of $G$, then $G=B^2[V]$ and $B$ is ${\cal C}(G)$-convex.
\item[\em (ii)] If $B=(X,Y,E_B)$ is $X$-convex, then $B^2[Y]$ is an interval graph. 
\end{itemize}
\end{lemma}
\begin{proof}

\noindent
(i): Let $G=(V,E_G)$ be an interval graph, and let $B=(V,{\cal C}(G),E_B)$ be the vertex-clique incidence bipartite graph of $G$. Since each $v\in V$ appears in the interval $\{Q_i\mid \ell(v)\le i\le r(v)\}$ in ${\cal C}(G)=(Q_1,\ldots, Q_q)$, $B$ is ${\cal C}(G)$-convex. Moreover, by Proposition~\ref{prop:subdivision-incidence}, $G=B^2[V]$. 

\smallskip\noindent
(ii): This is because $X$ admits a linear ordering such that, for each $y\in Y$, $N(y)$ is an interval in $X$. This collection is an interval representation of $B^2[Y]$ because $y$ and $y'$ are adjacent in $B^2[Y]$ if and only if $N(y)\cap N(y')\not=\emptyset$. 
\qed
\end{proof}

\begin{theorem}\label{thm:biconvex}
A graph is half-square of a biconvex bipartite graph if and only if it is a unit interval graph.
\end{theorem}
\begin{proof}
First, by Lemma~\ref{lem:convex} (ii), half-squares of biconvex bipartite graphs are interval graphs, and then by Lemma~\ref{lem:biconvex-claw}, half-squares of biconvex bipartite graphs are unit interval graphs.

Next we show that every unit interval graph is half-square of some biconvex bipartite graph. Let $G=(V,E_G)$ be a unit interval graph. Let $B=(V, {\cal C}(G), E_B)$ be the vertex-clique incidence bipartite graph of $G$. By Lemma~\ref{lem:convex} (i), $G=B^2[V]$ and $B$ is ${\cal C}(G)$-convex. We now are going to show that $B$ is $V$-convex, too.

Consider a linear order in ${\cal C}(G)$, ${\cal C}(G)=(Q_1,\ldots, Q_q)$, such that each $v\in V$ is contained in exactly the cliques $Q_i$, $\ell(v)\le i\le r(v)$. Let $v\in V$ be lexicographically sorted according $(\ell(v),r(v))$. 
We claim that $B$ is $V$-convex with respect to this ordering. Assume, by a contradiction, that some $Q_i$ has neighbors $v,u$ and non-neighbor $x$ with $v < x < u$ in the sorted list, say. In particular, $v, u$ belong to $Q_i$, but $x$ not; see also Figure~\ref{fig:biconvex}. 

\begin{figure}[ht]
\begin{center}
\begin{tikzpicture}[scale=.5]
\tikzstyle{vertex}=[draw,circle,inner sep=1.5pt]; 

\node (lv)   at (1,4) [label=above:$Q_{\ell(v)}$] {};
\node at (3,4) {$\ldots$};
\node (lx-1) at (5,4) [label=above:$Q_{\ell(x)-1}$] {}; 
\node (lx)   at (8,4) [label=above:$Q_{\ell(x)}$] {}; 
\node at (10,4) {$\ldots$};
\node (rx)   at (12,4) [label=above:$Q_{r(x)}$] {}; 
\node (rx+1) at (15,4) [label=above:$Q_{r(x)+1}$] {}; 
\node at (17,4) {$\ldots$};
\node (i)    at (19,4) [label=above:$Q_i$] {}; 

\node  at (1,3) {$v$}; 
\node  at (1,2) {$\vdots$};
\node  at (1,1) {$\vdots$};

\node  at (5,3) {$v$}; 
\node  at (5,2) {$y$};
\node  at (5,1) {$\vdots$};

\node  at (8,3) {$v$}; 
\node  at (8,2) {$x$};
\node  at (8,1) {$\vdots$};

\node  at (12,3) {$v$}; 
\node  at (12,2) {$x$};
\node  at (12,1) {$\vdots$};

\node  at (15,3) {$v$}; 
\node  at (15,2) {$z$};
\node  at (15,1) {$\vdots$};

\node  at (19,3) {$v$}; 
\node  at (19,2) {$u$};
\node  at (19,1) {$\vdots$};
\end{tikzpicture}
\end{center}
\caption{Assuming $v<x<u$, and $v,u\in Q_i$, but $x\not\in Q_i$.}\label{fig:biconvex}
\end{figure}
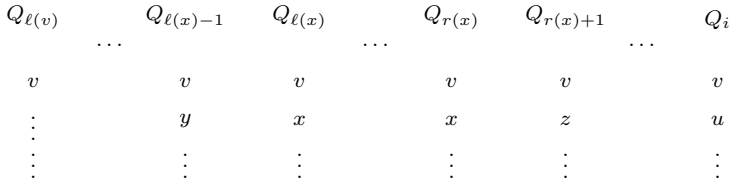

Since $x < u$ and $\ell(u)\le i$, we have $\ell(x) <i$. Since $x$ is not in $Q_i$, we therefore have   
\[\ell(x)\le r(x) < i.\] 
In particular, $r(x)+1\le i$. Since $v < x$ and $r(v) \ge i$, we have 
\[\ell(v) < \ell(x).\]  
Hence $\ell(x)-1 \ge 1$. Now, by the maximality of the cliques, there exists $y\in Q_{\ell(x)-1}$ with $r(y)=\ell(x)-1$ (hence $y$ is non-adjacent to $x$), and there exists $z\in Q_{r(x)+1}$ with $\ell(z)=r(x)+1$ (hence $z$ is non-adjacent to $x$ and $y$; note that $r(x)+1=i$ and $z=u$ are possible). But then $v, x, y$, and $z$ induce a $K_{1,3}$ in $G$, a contradiction.

Thus, we have seen that every unit interval graph is half-square of a biconvex bipartite graph.
\qed
\end{proof}

We next characterize half-squares of convex bipartite graphs as interval graphs. This is somehow surprising because the definition of being convex bipartite is asymmetric with respect to the two half-squares. 

\begin{theorem}\label{thm:convex}
A graph is a half-square of a convex bipartite graph if and only if it is an interval graph.
\end{theorem}
\begin{proof}
By Lemma~\ref{lem:convex} (i), interval graphs are half-squares of convex bipartite graphs. 
It remains to show that half-squares of convex bipartite graphs are interval graphs. Let $B=(X,Y,E_B)$ be an $X$-convex bipartite graph. By Lemma~\ref{lem:convex} (ii), $B^2[Y]$ is an interval graph. We now are going to show that $B^2[X]$ is an interval graph, too.

Let $B'=(X,Y',E_{B'})$ be obtained from $B$ by removing all vertices $y\in Y$ with $N_B(y)$ is properly contained in $N_B(y')$ for some $y'\in Y$. Clearly, $B^2[X]=B'^2[X]$. We show that $B'$ is $Y'$-convex, hence, by Lemma~\ref{lem:convex} (ii), $B^2[X]=B'^2[X]$ is an interval graph, as claimed. 
To this end, let $X=\{x_1,\ldots, x_n\}$ such that, for every $y\in Y'$, $N_{B'}(y)$ is an interval in $X$. (Recall that $B$, hence $B'$ is $X$-convex.) 
For $y\in Y'$ let $\text{left}(y)=\min\{i\mid x_i\in N_{B'}(y)\}$, and sort $y\in Y'$ increasing according $\text{left}(y)$. 
Then, for each $x\in X$, $N_{B'}(x)$ is an interval in $Y'$: Assume, by contradiction, that there is some $x\in X$ such that $N_{B'}(x)$ is not an interval in $Y'$. Let $y$ be a leftmost and $y'$ be a rightmost vertex in $N_{B'}(x)$. By the assumption, there is some $y''\in Y'\setminus N_{B'}(x)$ with $\text{left}(y) \leq \text{left}(y'') \leq  \text{left}(y')$. Then, as $N_{B'}(y), N_{B'}(y'')$ and $N_{B'}(y')$ are intervals in $X$, $N_{B'}(y'')$ must be a subset of $N_{B'}(y)$; see also Figure~\ref{fig:convex}. Since $x\in N_{B'}(y)$ but $x\not\in N_{B'}(y'')$, $N_{B'}(y'')$ must be a proper subset of $N_{B'}(y)$, contradicting to the fact that in $B'$, no such pair of vertices exists in $Y'$. Thus, $B'$ is $Y'$-convex. 

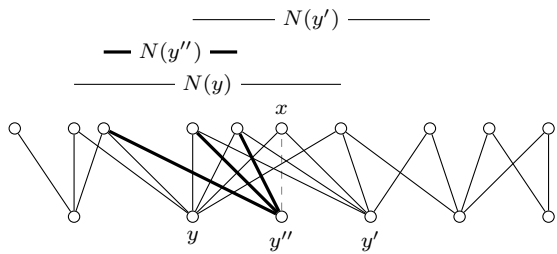
\begin{figure}[ht]
\begin{center}
\begin{tikzpicture}[scale=.39]
\tikzstyle{vertex}=[draw,circle,inner sep=1.5pt]; 

\node[vertex] (x0) at (-3,3) {};
\node[vertex] (x1) at (-1,3) {}; 
\node[vertex] (x2) at (0,3) {}; 
\node[vertex] (x3) at (3,3) {}; 
\node[vertex] (x4) at (4.5,3) {}; 
\node[vertex] (x) at  (6,3) [label=above:$x$] {}; 
\node[vertex] (x5) at (8,3) {}; 
\node[vertex] (x6) at (11,3) {};
\node[vertex] (x7) at (13,3) {};
\node[vertex] (x8) at (15,3) {};

\node[vertex]  (y0) at (-1,0) {};
\node[vertex]  (y) at (3,0) [label=below:$y$] {};
\node[vertex]  (y'') at (6,0)[label=below:$y''$] {};
\node[vertex]  (y') at (9,0) [label=below:$y'$] {};
\node[vertex]  (y1) at (12,0) {};
\node[vertex]  (y2) at (15,0) {};
\draw (y) -- (x1); \draw (y) -- (x); \draw (y) -- (x5);
\draw[dashed, 
      very thin, gray] (y'') -- (x);
\draw (y') -- (x3); \draw (y') -- (x); \draw (y') -- (x6);

\draw[very thick] (x2) -- (y'') -- (x4); \draw[very thick] (y'') -- (x3); 
\draw[thin] (y) -- (x2); \draw[thin] (y) -- (x3); \draw[thin] (y) -- (x4); 
\draw[thin] (y') -- (x4); \draw[thin] (y') -- (x5);  

\draw[thin] (y0) -- (x0); \draw[thin] (y0) -- (x1); \draw[thin] (y0) -- (x2);
\draw[thin] (x5) -- (y1) -- (x6); \draw[thin] (y1) -- (x7); \draw[thin] (y1) -- (x8); 
\draw[thin] (y2) -- (x7); \draw[thin] (y2) -- (x8);

\draw (-1,4.5) -- (8,4.5) node[pos=.5, fill=white] {$N(y)$};
\draw[very thick] (0,5.6) -- (4.5,5.6) node[pos=.5, fill=white] {$N(y'')$};
\draw (3,6.7) -- (11,6.7) node[pos=.5, fill=white] {$N(y')$};
\end{tikzpicture}
\end{center}
\caption{Assuming $\text{left}(y)\leq\text{left}(y'')\leq \text{left}(y')$, and $x$ is adjacent to $y$ and $y'$, but non-adjacent to $y''$.}\label{fig:convex}
\end{figure}

Note that $B'$ is indeed biconvex, hence, by Theorem~\ref{thm:biconvex}, $B^2[X]=B'^2[X]$ is even a unit interval graph.  
\qed
\end{proof}

Since (unit) interval graph with $n$ vertices and $m$ edges can be recognized in linear time $O(n+m)$ and all maximal cliques of an (unit) interval graph can be listed in the same time complexity (cf. \cite{Golumbic}), Theorems~\ref{thm:convex} and~\ref{thm:biconvex} imply:

\begin{corollary}
\textsc{Half-Square Of Convex} and \textsc{Half-Square Of Biconvex} can be solved in linear time. A (bi)convex bipartite half-root, if any, can be constructed in linear time. 
\end{corollary}


\section{Half-squares of chordal bipartite graphs}\label{sec:chordalbip}

In this section, we show that half-squares of chordal bipartite graphs are precisely the strongly chordal graphs. 
Recall that a graph is chordal if it has no induced cycle of length at least four. It is well-known (see, e.g.,~\cite{Golumbic,McKeyM,Spinrad}) that a graph $G=(V,E_G)$ is chordal if and only if it admits a tree representation, that is, there exists a tree $T$ such that, for each vertex $v\in V$, $T_v$ is a subtree of $T$ and two vertices in $G$ are adjacent if and only if the corresponding subtrees in $T$ intersect. Moreover, the vertices of $T$ can be taken as the maximal cliques of the chordal graph (a clique tree). 
Recall also that a graph is strongly chordal if it is chordal and has no induced $k$-sun, $k\ge 3$. Here, a $k$-sun consists of a stable set $\{s_1,\ldots, s_k\}$ and a clique $\{t_1,\ldots,t_k\}$ and edges $s_it_i$, $s_it_{i+1}$, $1\le i\le k$. (Indices are taken modulo $k$.) 
 
We first begin with the following fact.

\begin{lemma}\label{lem:1}
Let $B=(V,W,E_B)$ be a bipartite graph without induced $C_6$ and let $k\ge 3$. If $B^2[V]$ contains an induced $k$-sun, then $B$ contains an induced cycle of length $2k$.
\end{lemma}
\begin{proof} 
We first show that 
\[\text{every clique $Q$ in $B^2[V]$ stems from a star in $B$.}\]
 Suppose, by a contradiction, that there is some clique $Q$ in $B^2[V]$ such that, for any vertex $w\in W$, $Q\setminus N(w)\not=\emptyset$. Choose a vertex $w_1\in W$ with $Q':=Q\cap N(w_1)$ is maximal. Let $v_1\in Q\setminus N(w_1)$. Since $Q$ is a clique in $B^2[V]$, there is a vertex $w_2\in W$ adjacent to $v_1$ and some vertices in $Q'$. Choose such a vertex $w_2$ with $Q'\cap N(w_2)$ is maximal. 
By the choice of $w_1$, there is a vertex $v_2\in Q'\setminus N(w_2)$. Again, since $Q$ is a clique in $B^2[V]$, there is a vertex $w_3\in W$ adjacent to both $v_1$ and $v_2$. By the choice of $w_2$, there is a vertex $v_3\in Q'\cap N(w_2)$ non-adjacent to $w_3$. 
But then $w_1, v_2, w_3, v_1, w_2$ and $v_3$ induce a $C_6$ in $B$, a contradiction. Thus, there must be a vertex $w\in W$ such that $Q\subseteq N(w)$.
 
 Now, consider a $k$-sun in $B^2[V]$ with stable set $\{s_1,\ldots, s_k\}$, clique $\{t_1,\ldots,t_k\}$, and edges $s_it_i$, $s_it_{i+1}$, $1\le i\le k$. Let $w\in W$ such that $\{t_1,\ldots,t_k\}\subseteq N(w)$, and let $w_i\in W$ such that $\{s_i, t_i, t_{i+1}\}\subseteq N(w_i)$, $1\le i\le k$. Since $s_i$ is adjacent, in the $k$-sun, only to $t_i$ and $t_{i+1}$, we have in $B$, that $w_i$ is non-adjacent to $\{t_1,\ldots,t_k\}\setminus\{t_i,t_{i+1}\}$. That is, $t_1, w_1, t_2, w_2,\ldots, t_k, w_k$ induce a $C_{2k}$ in $B$.   
\qed
\end{proof}

\begin{theorem}\label{thm:chordalbip}
A graph is half-square of a chordal bipartite graph if and only if it is a strongly chordal graph.
\end{theorem}  
\begin{proof}
We first show that half-squares of chordal bipartite graphs are chordal. Let $B=(X,Y,E_B)$ be a chordal bipartite graph. It is known that $B$ is tree convex (\cite{JiangLWX,Lehel}). Thus, there is a tree $T=(X,E_T)$ such that, for each $y\in Y$, $N(y)$ induces a subtree in $T$. Then, for distinct vertices $y, y'\in Y$, $y$ and $y'$ are adjacent in $B^2[Y]$ if and only if $N(y)\cap N(y')\not=\emptyset$, and thus, $B^2[Y]$ has a tree representation, hence chordal. Now, by Lemma~\ref{lem:1}, $B^2[Y]$ cannot contain any sun $k$-sun, $k\ge 3$, showing that it is a strongly chordal graph. By symmetry, $B^2[X]$ is also strongly chordal. We have seen that half-squares of chordal bipartite graphs are strongly chordal graphs.

Next, let $G=(V,E_G)$ be a strongly chordal graph, and let $B=(V,{\cal C}(G), E_B)$ be the vertex-clique incidence bipartite graph of $G$. By Proposition~\ref{prop:subdivision-incidence}, $G=B^2[V]$. Moreover, it is well-known (\cite{Farber}) that $B$ is chordal bipartite. Thus, every strongly chordal graph is a half-square of some chordal bipartite graph, namely of its vertex-clique incidence bipartite graph.
\qed
\end{proof}

Testing if $G$ is strongly chordal can be done in $O(\min\{n^2,m\log n\})$ time 
(\cite{Farber,Lubiw1987,Spinrad}). Assuming $G$ is strongly chordal, 
all maximal cliques $Q_1$, \ldots, $Q_q$ of $G$ can be listed in linear time 
(cf. \cite{Golumbic,Spinrad}); note that $q\le n$. So, Theorem~\ref{thm:chordalbip} implies:

\begin{corollary}\label{cor:chordalbip}
\textsc{Half-Square Of Chordal Bipartite} can be solved in time $O(\min\{n^2,m\log n\})$, where $n$ and $m$ are the vertex and edge number of the input graph, respectively. A chordal bipartite half-root, if any, can be constructed in the same time complexity. 
\end{corollary}

In the rest of this section, we give another proof for a characterization of half-squares of trees found in~\cite{MnichRS}. Note that half-squares of trees generalize trees: Let $G=(V,E_G)$ be a tree. Then the subdivision $T=(V,E_G,E_T)$ of $G$ is also a tree and $G=T^2[V]$ (cf. Proposition~\ref{prop:subdivision-incidence}). Thus, trees are half-squares of some trees. But half-squares of trees are more general.

A {\em block graph\/} is one in which every maximal $2$-connected subgraph (a block) is a complete graph; equivalently, a block graph is a chordal graph without induced $K_4-e$, a $K_4$ minus an edge. In particular, trees are block graphs.

\begin{theorem}[\cite{MnichRS}]\label{thm:tree}
A graph is half-square of a tree if and only if it is a block graph.
\end{theorem}
\begin{proof}
Let $T=(V,W,E_T)$ be a tree. By symmetry, we only consider the half square $T^2[V]$. By Theorem~\ref{thm:chordalbip}, $T^2[V]$ is chordal. Since $T$ has no cycle, $T^2[V]$ cannot contain $K_4-e$ as an induced subgraph: for, if $\{a,b,c\}$ and $\{b,c,d\}$ are the two triangles of an induced $K_4-e$ in $T^2[V]$, then $\{a,b,c\}\subseteq N_T(w_1)$ and $\{b,c,d\}\subseteq N_T(w_2)$ for some $w_1, w_2\in W$, and $b, c, w_1, w_2$ would induce a $4$-cycle in $T$.  Thus, $T^2[V]$ is a block graph.

Conversely, let $G=(V,E_G)$ be a (connected) block graph. Then, as the maximal cliques in $G$ are exactly the blocks of $G$, the vertex-clique incidence bipartite graph $B(V,{\cal C}(G),E_B)$ of $G$ is a tree. Thus, block graphs are half-squares of trees.  
\qed
\end{proof}


\section{Conclusions}\label{sec:con}

Until recently, only half-squares of planar bipartite graphs (the map graphs) have been investigated, and the most considered problem is if it is possible to recognize these graphs faster and simpler than Thorup's $O(n^{120})$-time algorithm.  

In this paper, we initiate an investigation of half-squares of not necessarily planar bipartite graphs. We have shown the first NP-hardness result, namely that recognizing if a graph is half-square of a balanced bisplit graph is NP-complete. For classical subclasses of tree convex bipartite graphs such as star convex and star biconvex, convex and biconvex, and chordal bipartite graphs, we have given good structure characterizations for their half-squares. These structural results imply that half-squares of these restricted classes of tree convex bipartite graphs can be recognized efficiently. 

Recall that chordal bipartite graphs form a subclass of tree biconvex bipartite graphs (\cite{JiangLWX,Lehel}), and that half-squares of chordal bipartite graphs can be recognized in polynomial time (Corollary~\ref{cor:chordalbip}), while the complexity of recognizing half-squares of tree (bi)convex bipartite graphs is unknown. So, an obvious question is: what is the computational complexity of \textsc{Half-Square Of Tree (Bi)convex}?

\medskip\noindent
\textbf{Acknowledgment}.\, 
We thank Hannes Steffenhagen for his careful reading and very helpful remarks. We also thank one of the unknown referees for his/her helpful comments and suggestions, and for pointing out a gap in an earlier proof of Lemma~\ref{lem:1}.

\end{document}